\documentclass{article} % Include author names

\usepackage{amsmath}
\usepackage{amssymb}
\usepackage[comma,authoryear]{natbib}
\usepackage{graphicx}
\usepackage{url}
\usepackage[utf8]{inputenc} % allow utf-8 input
\usepackage[T1]{fontenc}    % use 8-bit T1 fonts
\usepackage{hyperref}       % hyperlinks
\usepackage{url}            % simple URL typesetting
\usepackage{booktabs}       % professional-quality tables
\usepackage{amsfonts}       % blackboard math symbols
\usepackage{microtype}      % microtypography
\usepackage{graphicx}
\usepackage{color}
\usepackage{amsthm, amsmath, amssymb}
\usepackage{subcaption}
\usepackage{caption}

\newtheorem{thm}{Theorem}
\newtheorem{lem}[thm]{Lemma}

\newtheorem{prop}[thm]{Proposition}

\theoremstyle{definition}

\theoremstyle{definition}
\newtheorem{mydef}[thm]{Definition}

\theoremstyle{definition}

\title{Property Testing for Differential Privacy}
\usepackage{times}
\usepackage{booktabs}       % professional-quality tables
\usepackage{algorithm}
\usepackage{algorithmic}
\usepackage{multicol}

\usepackage[margin=3.5cm, top=1.6in]{geometry}

 \author{Anna Gilbert \and Audra McMillan\footnote{Email: amcm@umich.edu}} %\Email{amcm@umich.edu}\\
 \date{\it Department of Mathematics \\ University of Michigan, Ann Arbor, MI 48109, USA}

\begin{document}

\maketitle

\begin{abstract}
We consider the problem of property testing for differential privacy: with black-box access to a purportedly private algorithm, can we verify its privacy guarantees? In particular, we show that \emph{any privacy guarantee that can be efficiently verified is also efficiently breakable} in the sense that there exist two databases between which we can efficiently distinguish. We give lower bounds on the query complexity of verifying pure differential privacy, approximate differential privacy, random pure differential privacy, and random approximate differential privacy. We also give algorithmic upper bounds. The lower bounds obtained in the work are infeasible for the scale of parameters that are typically considered reasonable in the differential privacy literature, even when we suppose that the verifier has access to an (untrusted) description of the algorithm. A central message of this work is that verifying privacy requires compromise by either the verifier or the algorithm owner. Either the verifier has to be satisfied with a weak privacy guarantee, or the algorithm owner has to compromise on side information or access to the algorithm.
\end{abstract}

%\begin{keywords}
%Property testing, differential privacy, black-box testing
%\end{keywords}

\section{Introduction}

Recently, differential privacy (DP) has gained traction outside of theoretical research as several companies (Google, Apple, Microsoft, Census, etc.) have announced deployment of large-scale differentially private mechanisms \citep{RAPPOR, appleDPwhatsnew, Abowd:2017, Bolin:2017}. This use of DP, while exciting, might be construed as a marketing tool used to encourage privacy-aware consumers to release more of their sensitive data to the company. In addition, the software behind the deployment of DP is typically proprietary since it ostensibly provides commercial advantage. This raises the question: with limited access to the software, can we verify the privacy guarantees of purportedly DP algorithms?

Suppose there exists some randomised algorithm $\mathcal{A}$ that is claimed to be $\Xi$- differentially private and we are given query access to $\mathcal{A}$. That is, the domain of $\mathcal{A}$ is the set of databases and we have the power to choose a database $D$ and obtain a (randomised) response $\mathcal{A}(D)$. How many queries are required to verify the privacy guarantee? We formulate this problem in the property testing framework for pure DP, approximate DP, random pure DP, and random approximate DP.

\begin{mydef}[Property testing with side information]\label{proptesting}
A property testing algorithm with query complexity $q$, proximity parameter $\alpha$, privacy parameters $\Xi$ and side information $S$, makes $q$ queries to the black-box and:
\begin{enumerate}
\item (Completeness) ACCEPTS with probability at least $2/3$ if $\mathcal{A}$ is $\Xi$-private and $S$ is accurate.
\item (Soundness) REJECTS with probability at least $2/3$ if $\mathcal{A}$ is $\alpha$-far from being $\Xi$-private.
\end{enumerate}
\end{mydef}

In this early stage of commercial DP algorithms, approaches to transparency have been varied. For some algorithms, like Google's RAPPOR, a full description of the algorithm has been released \citep{RAPPOR}. On the other hand, while Apple has released a white paper \citep{Apple:2017} and a patent \citep{Thakurta:2017}, there are still many questions about their exact implementations. We focus on the two extreme settings: when we are given \emph{no information} about the black-box (except the domain and range), and the \emph{full information} setting where we have an untrusted full description of the algorithm $\mathcal{A}$. A similar formulation of DP in the property testing framework was first introduced in \cite{Dixit:2013}, who consider testing for DP given oracle access to the probability density functions on outputs. \cite{Dixit:2013} reduce this version of the problem to testing the Lipschitz property of functions and make progress on this more general problem. 

Both settings we consider, full information and no information, are subject to fundamental limitations. We first show that \emph{verifying} privacy is at least as difficult as \emph{breaking} privacy, even in the full information setting. That is, suppose $r$ samples are sufficient to verify that an algorithm is $\Xi$-private. Then Theorem~\ref{verifydistinguish} implies that for every algorithm that is not $\Xi$-private, there exists some pair of neighbouring databases $D$ and $D'$ such that $r$ samples from $\mathcal{A}(D)$ is enough to distinguish between $D$ and $D'$. Differential privacy is designed so that this latter problem requires a large number of samples. This connection has the unfortunate implication that verifiability and privacy are directly at odds: \emph{if a privacy guarantee is efficiently verifiable, then it mustn't be a strong privacy guarantee}. 

For the remainder of this work we restrict to discrete distributions on $[n]$. Our upper and lower bounds in each setting are contained in Table \ref{results}. We rule out sublinear verification of privacy in every case except verifying approximate differential privacy in the full information setting. That is, for all other definitions of privacy, the query complexity for property testing of privacy is $\Omega(n)$.

Each privacy notion we consider is a relaxation of pure differential privacy. Generally, the privacy is relaxed in one of two ways: either privacy loss is allowed to occur on unlikely \emph{outputs}, or privacy loss is allowed to occur on unlikely \emph{inputs}. The results in Theorem \ref{checkpairs} and the lower bounds in Table \ref{results} imply that for efficient verification, we need to relax in \emph{both} directions. That is, random approximate DP is the only efficiently verifiable privacy notion in the no information setting. Even then, we need about $1/\delta^2$ queries \emph{per database} to verify $(\epsilon, \delta)$-approximate differential privacy. Theorem~\ref{aDPUB} shows that random approximate DP can be verified in (roughly) $O(\frac{4n(1+e^{2\epsilon})\log(1/\gamma)}{\gamma\delta^2})$ samples, where (roughly) $\delta$ and $\gamma$ are the probabilities of choosing a disclosive output or input, respectively. This means verification is efficient if $\delta$ and $\gamma$ are small but not too small. This may seem insufficient to those familiar with DP, where common wisdom decrees that $\delta$ and $\gamma$ should be small enough that this query complexity is infeasibly large. 

There have been several other relaxations of pure differential privacy proposed in the literature, chief among them R\'enyi DP \citep{Mironov:2017} and concentrated DP \citep{Dwork:2016}. These relaxations find various ways to sacrifice privacy, with a view towards allowing a strictly broader class of algorithms to be implemented. Similar to pure DP, R\'enyi and concentrated DP have the property that two distributions $P$ and $Q$ can be close in TV distance while the pair $(P,Q)$ has infinite privacy parameters. Thus, many of the results for pure differential privacy in this work can be easily extended to R\'enyi and concentrated DP. We leave these out of our discussion for brevity.

\begin{table}[t]
\caption{Per Database Query complexity bounds for property testing of privacy}
\label{results}
\vskip 0.15in
\begin{center}
\begin{small}
%\begin{sc}
\begin{tabular}{lccr}
\toprule
&No Information&Full information\\
\hline
\hline
pDP&Unverifiable [Theorem \ref{nogopure}]& $\Omega\left(\frac{1}{\beta\alpha^2}\right)$ [Theorem \ref{FILB}] \\
&&$O\left(\frac{\ln n}{\alpha^2\beta^2}\right)$ [Theorem \ref{FIpureUB}]\\
\hline
aDP& $\Omega(\max\{n^{1-o(1)}, \frac{1}{\alpha^2}\})$ [Theorem \ref{aDPLB}]& $O\left(\frac{\sqrt{n}}{\alpha^2}\right)$[Theorem \ref{FIaDP}]\\
& $O(\frac{n}{\alpha^2})$ [Theorem \ref{aDPUB}]&\\
\bottomrule
\end{tabular}
%\end{sc}
\end{small}
\end{center}
\vskip -0.1in
\end{table}

One might hope to obtain significantly lower query complexity if the property tester algorithm is given side information, even if the side information is untrusted. We find that this is true for both approximate DP and pure DP, if we allow the query complexity to depend on the side information. A randomised algorithm $\mathcal{A}$ can be abstracted as a set of distributions $\{P_D\}$ where $\mathcal{A}(D)\sim P_D$. We obtain a sublinear verifier for approximate DP. For pure DP, we find the quantity that controls the query complexity is \[\beta = \inf_D\min_i P_D(i),\] the minimum value of the collection of distributions. If $\beta$ is large then efficient verification is possible: verifying that the pure differential privacy parameter is less than $\epsilon+\alpha$ requires $O(\frac{\ln n}{\alpha^2\beta^2})$ queries of each database (Theorem \ref{FIpureUB}). Note that this is not sublinear since $\beta\le1/n$ and if $\beta=0$ then we have no improvement on the no information setting. However, for reasonable $\beta$, this is a considerable improvement on the no information lower bounds and may be efficient for reasonable $n$.

A central theme of this work is that verifying the privacy guarantees that corporations (or any entity entrusted with private data) claim requires compromise by either the verifier or algorithm owner. If the verifier is satisfied with only a weak privacy guarantee (random approximate DP with $\delta$ and $\gamma$ small but not extremely small), then they can achieve this with no side information from the algorithm owner. If the company is willing to compromise by providing information about the algorithm up-front, then much stronger privacy guarantees can be verified. Given this level of transparency, one might be tempted to suggest that the company provide source code instead. While verifying privacy given source code is an important and active area of research, there are many scenarios where the source code itself is proprietary. We have already seen instances where companies have been willing to provide detailed descriptions of their algorithms. In the full information case, we obtain our lowest sample complexity algorithms, including a sublinear algorithm for verifying approximate differential privacy.

This paper proceeds as follows: we start by defining property testing for privacy in Section \ref{background}. We then proceed to the main contributions of this work:
\begin{itemize}
\item Verifying privacy is as hard as breaking privacy (Section \ref{distinguishsec}).
\item In the no information setting, verifying pure differential privacy is impossible while there is a finite query complexity property tester for approximate differential privacy (Section \ref{noinfo}).
\item If $\beta>0$, then finite query complexity property testers exist for pure differential privacy in the full information setting (Section \ref{fullinfo}).
\item A sublinear property tester exists for approximate differential privacy in the full information setting.
\end{itemize}
The main lower bounds and algorithmic upper bounds in this paper are summarized in Table~\ref{results}. 

\section{Background and Problem Formulation}\label{background}

A database is a vector $D$ in $\mathbb{Z}^{|\Omega|}$ for some data universe $\Omega$. That is, if $\omega\in\Omega$, $D_{\omega}$ is the number of copies of $\omega$ in the database. We call two databases $D$, $D'$ \emph{neighbouring} if they differ on a single data point, that is $\|D-D'\|_1=1$. For a randomised algorithm $\mathcal{A}$ and database $D$, we use $\mathcal{A}(D)$ to denote the output and $P_D$ to denote the distribution of $\mathcal{A}(D)$. We will often prefer to view an algorithm as simply a collection of distributions $\{P_D\;|\; D\in\mathbb{Z}^{|\Omega|}\}$. We will only consider discrete distributions in this paper, so $P_D$ is a discrete distribution on $[n]=\{1,\cdots,n\}$. For a distribution $P$, $P^r$ represents $r$ independent copies of $P$.

For much of this paper we will consider algorithms that accept only two databases as input. We use the notation $\mathcal{A} = (P_0, P_1)$ to denote such an algorithm that accepts only two databases $0$ and $1$ as input, and $\mathcal{A}(0)\sim P_0$ and $\mathcal{A}(1)\sim P_1$. The databases 0 and 1 are assumed to be neighbouring.

The privacy notions we discuss will all center around the idea that $P_D$ and $P_{D'}$ should be close for neighbouring databases $D$ and $D'$. As such, we will deal with many measures of closeness between distributions. We collect these definitions for ease of reference.

\begin{mydef} Let $P$ and $Q$ be two distributions. 
\begin{itemize}
\item \emph{(Max divergence)} $D_{\infty}(P,Q) = \sup_{E}\ln\frac{P(E)}{Q(E)}$.
\item \emph{($\delta$-approximate max divergence)} $D_{\infty}^{\delta}(P,Q) = \sup_{E \text{ s.t. } P(E)\ge\delta}\ln\frac{P(E)-\delta}{Q(E)}.$
%\item \emph{(R\'enyi divergence of order $\beta$)}  $D_{\beta}(P\|Q) = \frac{1}{\beta-1}\ln\mathbb{E}_{x\sim Q}\left(\frac{P(x)}{Q(x)}\right)^{\beta}.$
\item \emph{(KL divergence)} $D_{KL}(P\|Q) = \int_R P(x)\ln \frac{P(x)}{Q(x)} dx$.
\item \emph{(Total Variance (TV) distance)} $\|P-Q\|_{\text{TV}} = \sup_E |P(E)-Q(E)|$.
\end{itemize}
where the $\sup_E$ is the supremum over all events $E$ in the outcome space.
\end{mydef}

\subsection{Privacy Definitions}\label{sec:privacy}

Pure differential privacy is the gold standard for privacy-preserving data analysis. However, it is a very strong definition and as a result, many relaxations of it have gained traction as the work on differential privacy evolves. These relaxations find various ways to sacrifice privacy, with a view towards allowing a strictly broader class of algorithms to be implemented. Since these definitions are becoming standard, we give only a cursory introduction in this section. An introduction can be found in \cite{Dwork:2014} and more in depth surveys can be found in \cite{Vadhan:2016, Dwork:2008, Ji:2014}. %The definitions presented here are not an exhaustive list of privacy notions. A notable omission is \emph{concentrated differential privacy}. Similar results to those true for pure DP and R\'enyi dDP presented in this paper are true for concentrated DP. We omit these because we are focussed on discrete distributions for which the concentrated DP is undefined.

The idea is simple; suppose the adversary has narrowed the list of possible databases down to neighbouring databases $D$ and $D'$. Any output the adversary sees is almost equally as likely to have arisen from $P_D$ or $P_{D'}$. Thus, the adversary gains almost no information that helps them distinguish between $D$ and $D'$. 

\begin{mydef}[Data Distribution Independent Privacy Definitions]
A randomised algorithm $\mathcal{A}$ is 
\begin{itemize}
\item $\epsilon$-\emph{pure differentially private (pDP)} if $\sup_{D,D'} D_{\infty}(P_D,P_{D'})\le \epsilon$.
\item $(\epsilon, \delta)$-\emph{approximate differentially private (aDP)} if $\sup_{D,D'} D_{\infty}^{\delta}(P_D,P_{D'})\le \epsilon$.
%\item $\epsilon$-\emph{R\'enyi differentially private of order $\beta$ ($\beta$RDP)} if $\sup_{D,D'} D_{\beta}(P_D\|P_{D'})\le \epsilon.$
\end{itemize}
where the supremums are over all pairs of neigbouring databases $D$ and $D'$.
\end{mydef}

Note that $\epsilon$-pDP is exactly $(\epsilon , 0)$-aDP. The parameter $\delta$ can be thought of as our probability of failing to preserve privacy. To see this, suppose the distributions $P_D$ output 0 with probability $1-\delta$, and a unique identifier for the database $D$ with probability $\delta$. Then this algorithm is $(0,\delta)$-DP. Thus, we typically want $\delta$ to be small enough that we can almost guarantee that we will not observe this difference in the distributions. In contrast, while it is desirable to have $\epsilon$ small, a larger $\epsilon$ still gives meaningful guarantees (\cite{Dwork:2011}). Typically one should think of $\delta$ as \emph{extremely} small, $\delta\approx10^{-8}$, and $\epsilon$ as \emph{quite} small, $\epsilon\approx 0.1$. The larger $\beta$ is, the more \emph{private} the algorithm is. Unlike the other parameters we will treat $\beta$ as a fixed part of the definition, rather than a variable like $\epsilon$. 

Let $\mathcal{D}$ be a distribution on the data universe $\Omega$. For a database $D$ and datapoint $z$, let $[D_{-1}, z]$ denote the neighbouring database where the first datapoint of $D$ is replaced by $z$. 

\begin{mydef}[Data Distribution Dependent Privacy Definitons]
An algorithm $\mathcal{A}$ is 
\begin{itemize}
\item \emph{$(\epsilon, \gamma)$-Random pure differentially private (RpDP)} if $\mathbb{P}\left(D_{\infty}(P_D,P_{[D_{-1}, z]})\le\epsilon\right)\ge 1-\gamma$.
\item \emph{$(\epsilon, \delta, \gamma)$-Random approximate differentially private (RADP)} if $\mathbb{P}\left(D_{\infty}^{\delta}(P_D,P_{[D_{-1}, z]})\le\epsilon\right)~\ge~1~-~\gamma$.
%\item \emph{$\epsilon$-On-Average KL-privacy (OnAvKLP)} if $\mathbb{E}[ \text{D}_{\text{KL}}(P_{[D_{-1}, z]}\|P_{D})]\le \epsilon$
\end{itemize}
where the probabilities in RpDP and RaDP are over $D\sim\mathcal{D}^n, z~\sim~\mathcal{D}$.
\end{mydef}

Similar to $\delta$, $\gamma$ represents the probability of catastrophic failure in privacy. Therefore, we require that $\gamma$ is small enough that this event is extremely unlikely to occur.

\subsection{Problem Formulation}\label{probform}

Our goal is to answer the question \emph{given these privacy parameters, is the algorithm $\mathcal{A}$ at least $\Xi$-private?} where $\Xi$ is an appropriate privacy parameter. A property testing algorithm, which outputs ACCEPT or REJECT, answers this question if it ACCEPTS whenever $\mathcal{A}$ is $\Xi$-private, and \emph{only} ACCEPTS if the algorithm $\mathcal{A}$ is close to being $\Xi$-private. A tester with side information may also REJECT simply because the side information is inaccurate.

We say that $\mathcal{A}$ is $\alpha$-far from being $\Xi$-private if $\min_{\Xi'} \|\Xi'-\Xi\|>\alpha$, where the minimum is over all $\Xi'$ such that $\mathcal{A}$ is $\Xi'$-private. The metrics used for each form of privacy are contained in Table \ref{metrictable}. 
We introduce the scalar $\lambda$ to penalise deviation in one parameter more than deviation in another parameter. For example, it is much worse to mistake a $(0,0.1)$-RpDP algorithm for $(0,0)$-RpDP than it is to mistake a $(0.1,0)$-RpDP algorithm for $(0,0)$-RpDP. We leave the question of \emph{how much worse} as a parameter of the problem. However, we give the general guideline that if we want an $\alpha$ error to be tolerable in both $\epsilon$ and $\gamma$ then $\lambda\approx \frac{\epsilon}{\gamma}$, which may be large, is an appropriate choice. 

\begin{table}
\caption{Privacy notions, parameters and metrics.}
\label{metrictable}
\vskip 0.15in
\begin{center}
\begin{small}
\begin{tabular}{lccr}
\toprule
 Privacy Notion & $\Xi$ & $\|\Xi-\Xi'\|$ \\
 \hline
 \hline
 pDP & $\epsilon$ & $|\epsilon-\epsilon'|$ \\  
 \hline
 aDP & $(\epsilon, \delta)$ & $|\delta_{\epsilon}-\delta'_{\epsilon}|$\\
 \hline
 RpDP & $(\epsilon, \gamma)$ & $\min\{|\epsilon-\epsilon'|, \lambda|\gamma-\gamma'|\}$\\
 \hline
 RaDP& $(\epsilon, \delta, \gamma)$ & $\min\{|\delta_{\epsilon}-\delta'_{\epsilon}|, \lambda|\gamma-\gamma'|\}$\\
\bottomrule
\end{tabular}
\end{small}
\end{center}
\vskip -0.1in
\end{table}

The formal definition of a property tester with side information was given in Definition~\ref{proptesting}. A \emph{no information} property tester is the special case when $S=\emptyset$. A \emph{full information} property tester is the special case when $S=\{Q_D \}$ contains a distribution $Q_D$ for each database $D$. 
We use $Q_D$ to denote the distribution on outputs presented in the side information and $P_D$ to denote the true distribution on outputs of the algorithm being tested.
For $\alpha>0$ and privacy parameter $\Xi$, a full information (FI) property tester for this problem satisfies:
\begin{enumerate}
\item (Completeness) Accepts with probability at least $2/3$ if the algorithm is $\Xi$-private and $P_D~=~Q_D$ for all $D$.
\item (Soundness) Rejects with probability at least $2/3$ if the algorithm is $\alpha$-far from being $\Xi$-private.
\end{enumerate}

We only force the property tester to ACCEPT if the side information is \emph{exactly} accurate ($P_D=Q_D$). It is an interesting question to consider a property tester that is forced to ACCEPT if the side-information is \emph{close} to accurate, for example in TV-distance. We do not consider this in this work as being close in TV-distance does not imply closeness of privacy parameters.

For a database $D$, we will refer to the process of obtaining a sample from $P_D$ as \emph{querying the black-box}.  It will usually be necessary to input each database into the black-box multiple times. We will use $m$ to denote the number of \emph{unique} databases that are queries to the black-box and $r$ to denote the number of times each database is input. We will only consider algorithms where the number of samples from $P_D$ for each input database is $r$, so our query complexity is $mr$ for each algorithm. Our aim is verify the privacy parameters using as few queries as possible.

\subsection{Related Work}

This work connects to two main bodies of literature. There are several works on verifying privacy with different access models that share the same motivation as this work. In terms of techniques, our work is most closely linked to recent work on property testing of distributions.

Testing DP in the property testing framework was first considered in \cite{Dixit:2013}. The access model in this paper is different to ours but their goal is similar. Recent work by \cite{Ding:2018} studies privacy verification from a hypothesis testing perspective. They design a privacy verification algorithm which aims to find violations of the privacy guarantee. Their algorithm provides promising experimental results in non-adversarial settings (when the privacy guarantee is frequently violated), although they provide no theoretical guarantees. 

Several algorithms and tools have been proposed for formal verification of the DP guarantee of an algorithm \citep{Barthe:2014, Roy:2010, Reed:2010, Gaboardi:2013, Tschantz:2011}. Much of this work focuses on verifying privacy given access to a description of the algorithm. There is a line of work \citep{Barthe:2014, Roy:2010, Reed:2010, Gaboardi:2013, Tschantz:2011, Barthe:2012, Barthe:2013, McSherry:2009} using logical arguments (type systems, I/O automata, Hoare logic, etc.) to verify privacy. These tools are aimed at automatic (or simplified) verification of privacy of source code. There is another related line of work where the central problem is testing software for privacy leaks. This work focuses on \emph{blatant} privacy leaks, such as a smart phone application surreptitiously leaking a user's email \citep{Jung:2008, Enck:2010, Fan:2012}. %We are unaware of any work verifying DP assuming only black-box access to the private algorithm. 

Given sample access to two distributions $P$ and $Q$ and a distance measure $d(\cdot,\cdot)$, the question of distinguishing between $d(P,Q)\le a$ and $d(P,Q)\ge b$ is called \emph{tolerant property testing}. This question is closely related to the question of whether $\mathcal{A}=(P,Q)$ is private. There is a large body of work exploring lower bounds and algorithmic upper bounds for tolerant testing using standard distances (TV, KL, $\chi^2$, etc.) with both $a=0$ and $a>0$ \citep{Kamath:2018, Paninski:2008, Batu:2013, Jayadev:2015, Valiant:2014}. In our work, we draw most directly from the techniques of \cite{Valiant:2011}. 

%\section{Technical Contributions}

\section{Lower Bounds via Distinguishability}\label{distinguishsec}

We now turn to examining the fundamental limitations of property testing for privacy. We find that even in the full information setting, the query complexity to verifying privacy is lower bounded by the number of queries required to distinguish between two possible inputs. We expect the latter to increase with the strength of the privacy guarantee.

\begin{mydef}
Databases $D$ and $D'$ are $r$-\emph{distinguishable under} $\mathcal{A}$ if there exists a testing algorithm such that given a description of $\mathcal{A}$ and $x\sim P_{D''}^r$ where $D''\in\{D,D'\}$, it accepts with probability at least $2/3$ if $D''=D$ and rejects with probability at least 2/3 if $D''=D'$.
\end{mydef}

The following theorem says that the per database query complexity of a privacy property testing algorithm is lower bounded by the minimal $r$ such that two neighbouring databases are $r$-distinguishable under $\mathcal{A}$. Recall that we use the notation $\mathcal{A} = (P_0, P_1)$ to denote an algorithm that accepts only two databases $0$ and $1$ as input, and $\mathcal{A}(0)\sim P_0$ and $\mathcal{A}(1)\sim P_1$. The databases 0 and 1 are assumed to be neighbouring.

\begin{thm} \label{verifydistinguish}
Consider any privacy definition, privacy parameter $\Xi$, and let $\alpha>0$. Suppose there exists a $\Xi$-privacy property tester with proximity parameter $\alpha$ and (per database) query complexity $r$. Let $\mathcal{A}$ be an algorithm that is $\alpha$-far from $\Xi$-private. If the privacy notion is
\begin{itemize}
\item pDP or aDP then there exists a pair of neighbouring databases that are $r$-distinguishable under $\mathcal{A}$.
\item RpDP or RaDP and $\Xi = (\epsilon, \delta, \gamma)$, then a randomly sampled pair of neighbouring databases has probability at least $\gamma+\frac{\alpha}{\lambda}$ of being $r$-distinguishable.
\end{itemize}
\end{thm}

A major reason that DP has gained traction is that it is preserved even if the (randomised) algorithm is repeated. %(Lemma \ref{composition}). 
That is, if $k>0$ and $(P_D, P_{D'})$ is private, then $(P_D^k, P_{D'}^k)$ is private with slightly worse privacy parameters. %The rate of decay of the privacy parameters (in $k$) varies with the notion of privacy. 
Typically we want the privacy parameters to start small enough that $k$ has to be quite large before any pair of neighbouring databases can be distinguished between using the output. If the algorithm is known and trusted then distinguishability may be possible with a feasible number of samples, for example by mean estimation (Laplacian distribution, etc.). This is no longer necessarily the case when make no assumptions on form of the distributions $P_D$. In fact, many of our proofs in the following sections proceed by finding two distribution $P$ and $Q$ such that $(P,Q)$ has high privacy parameters but it is still difficult to distinguish between $P$ and $Q$ (for example because they only differ on a set with small measure). We consider this setting because we are considering an untrusted, possibly adversarial, algorithm owner. In the future we would like to explore assumptions that can be placed on the class of distributions that may lower the sample complexity of distinguishability.

\begin{proof}
We start with pDP or aDP and suppose such a $\Xi$-privacy property testing algorithm exists. Let $\mathcal{A}$ be an algorithm that is $\alpha$-far from $\Xi$-private. Since the privacy parameter is defined as a maximum over all neighbouring databases, there exists a pair of databases $D$ and $D'$ such that $(P_D,P_{D'})$ has the same privacy parameter as $\mathcal{A}$. We can design a tester algorithm that distinguishes between $D$ and $D'$ as follows: given input $x\sim P_{D''}^r$, first sample $y\sim P_D^r$. Then run the privacy property testing algorithm on $\mathcal{B} = (P_{D''}, P_D)$ with sample $(x,y)$. If $D''=D$ then $\mathcal{B}$ is 0-DP, so the property tester will accept with probability at least 2/3. If $D''=D'$ then $\mathcal{B}$ is $\alpha$-far from from $\Xi$-private so the property tester will reject with probability at least 2/3. 

Finally, suppose such a $(\epsilon, \gamma)$-RpDP property testing algorithm exists. Let $\mathcal{A}$ be an algorithm that is $\alpha$-far from $(\epsilon, \gamma)$-private so that, in particular, $\mathcal{A}$ is not $(\epsilon+\alpha,\gamma+\frac{\alpha}{\lambda})$-RpDP. Thus, if we randomly sample a pair of neighbouring databases $D$ and $D'$, with probability $\gamma+\frac{\alpha}{\lambda}$, $(P_D, P_{D'})$ is not $\epsilon+\alpha$-pDP. The remainder of the proof proceeds as above by noticing that the algorithm $(P_D,P_{D'})$ is $\alpha$-far from $(\epsilon, \gamma)$-RpDP and $(P_D, P_D)$ is $(\epsilon, \gamma)$-RpDP. The proof of almost identical for RaDP.
\end{proof}

\section{Restriction to Two Distribution Setting}

Differential privacy is an inherently local property. That is, verifying that $\mathcal{A}$ is $\Xi$-private means verifying that $(P_D,P_{D'})$ is $\Xi$-private, either always or with high probability, for pairs of neighbouring databases $D$ and $D'$. We refer to the problem of determining whether a pair of distributions $(P_0,P_1)$ satisfies $\Xi$-privacy as the \emph{two database setting}. We argue in this section that the hard part of privacy property testing is the two database setting. For this reason, from Section \ref{noinfo} onwards, we only consider the two database setting. Recall that we use the notation $\mathcal{A} = (P_0, P_1)$ to denote an algorithm that accepts only two databases $0$ and $1$ as input, and $\mathcal{A}(0)\sim P_0$ and $\mathcal{A}(1)\sim P_1$. The databases 0 and 1 are assumed to be neighbouring.

An algorithm is non-adaptive if it chooses $m$ pairs of distributions and queries the blackbox with each database $r$ times. It does not choose its queries adaptively. The following is a non-adaptive algorithm for converting a tester in the two database setting to a random privacy setting.

\iffalse
\begin{proof}
Algorithm \ref{alg:random} is a privacy property tester with $m<\frac{1}{?}$ 
Let $\mathbb{P}[none of the samples contain something in S] = (1-\gamma-\alpha)^m\le 1/3$, which means $m\ge \log(1/3)/\log(1-\gamma-\alpha)\sim 1/(\gamma+\alpha)$.
\end{proof}
\fi

\begin{algorithm}[t]
   \caption{Random-privacy Property Tester}
   \label{alg:random}
\begin{algorithmic}
   \STATE {\bfseries Input:} A two distribution property tester $T$, $\alpha, \gamma>0$, a data distribution $\mathcal{D}$
   \FOR{$i=1:m$}
   \STATE Sample $(D,D')$ neighbours from $\mathcal{D}$.
   \FOR{$j=1:\log(2\lambda/\alpha)$}
   \STATE $x_{ij} = 1$ if $T(P_D,P_{D'})$ REJECTS
   \ENDFOR
   \STATE $x_i = \lfloor \frac{1}{2}+\frac{1}{\log(2\lambda/\alpha)}\sum_{j=1}^{\log(2\lambda/\alpha)}x_{ij} \rfloor$
   \ENDFOR
   \STATE $y = \frac{1}{m}\sum_{i=1}^m x_i$
   \IF{$y\le \gamma+\frac{\alpha}{\lambda}$}
   \STATE {\bfseries Output:} ACCEPT
   \ELSE
   \STATE {\bfseries Output:} REJECT
   \ENDIF
\end{algorithmic}
\end{algorithm}

\begin{thm}[Conversion to random privacy tester]\label{conversion}
If there exists a $\Xi$-privacy property tester for the two database setting with query complexity $r$ per database and proximity parameter $\alpha$, then there exists a privacy property tester for $(\Xi, \gamma)$-random privacy with proximity parameter $2\alpha$ and query complexity \[O\left(r\log\left(\frac{2\lambda}{\alpha}\right)\frac{(\alpha/\lambda+\gamma)^2+\alpha/\lambda}{(\alpha/\lambda)^2}\right).\]
\end{thm}

\begin{proof} The conversion is given in Algorithm \ref{alg:random}.
We first prove completeness. Suppose $\mathcal{A}$ is $(\gamma, \Xi)$-random private. Let \[S=\{(D,D')\;|\; D, D' \text{ are neighbours and } (P_D,P_{D'}) \text{ is $\Xi$-private}\}\]  so $1-\gamma\le \mathbb{P}(S)\le 1$. Our goal is to estimate $\mathbb{P}(S)$ using the empirical estimate given by $\frac{1}{m}\sum_{i=1}^m x_i$. We perform the property tester $\log(\lambda/\alpha)$ times on the pair $(P_D,P_{D'})$ to reduce the failure probability from $1/3$ to $\frac{\alpha}{2\lambda}$ so
\begin{align*}
\mathbb{E}[x_i] &= \mathbb{P}(x_i=1\;|\; (D,D')\in S)\mathbb{P}(S)+\mathbb{P}(x_i=1\;|\; (D,D')\notin S)\mathbb{P}(S^c)\le \frac{\alpha}{2\lambda}+\gamma.
\end{align*}
Now, 
\[\mathbb{P}\left(y\ge \gamma+\frac{\alpha}{\lambda}\right)\le \mathbb{P}\left(y-\mathbb{E}[y]\ge \frac{\alpha}{2\lambda}\right)\le e^{\frac{-m(\alpha/\lambda)^2}{2(\frac{\alpha}{2\lambda}+\gamma)^2+2(\alpha/\lambda)}}\le\frac{1}{3}\]
where the first inequality follows from Bernstein's inequality \citep{Sridharan:2018}. 
Therefore, Algorithm \ref{alg:random} ACCEPTS with probability at least 2/3. To prove soundness suppose $\mathcal{A}$ is $2\alpha$-far from $\Xi$-private. Let  \[S=\{(D,D')\;|\; D, D' \text{ are neighbours and } (P_D,P_{D'}) \text{ is $2\alpha$-far from $\Xi$-private}\}\] so $1\ge \mathbb{P}(S)\ge \gamma+2\alpha/\lambda$. Then $\mathbb{E}[x_i]\ge \left(1-\frac{\alpha}{2\lambda}\right)\left(\gamma+2\frac{\alpha}{\lambda}\right) \ge \gamma+\frac{\alpha}{\lambda}+\frac{\alpha}{2\lambda}.$ Therefore as above, \[\mathbb{P}\left(y\le \gamma+\frac{\alpha}{\lambda}\right)\le \mathbb{P}\left(y-\mathbb{E}[y]\le \frac{-\alpha}{2\lambda}\right)\le \frac{1}{3}.\] So, Algorithm \ref{alg:random} REJECTS with probability at least 2/3.
\end{proof}

Notice that if $\gamma\approx \frac{\alpha}{\lambda}$ then the query complexity is approximately $ r\log(\frac{\lambda}{\alpha})\frac{\lambda}{\alpha}\approx \frac{r\log(\frac{1}{\gamma})}{\gamma}$. One shortcoming of the conversion algorithm in Theorem \ref{conversion} is that we need to know the data distribution $\mathcal{D}$. We can relax to an approximation $\mathcal{D}'$ that is close in TV-distance, but it is not difficult to see that $\|\mathcal{D}-\mathcal{D'}\|_1\le \frac{\alpha}{\lambda}$ is necessary. 

\begin{thm}[Lower bound]\label{checkpairs}
Let $\gamma, \alpha>0$. Let $r$ be a lower bound on the query complexity in the two distribution settting.
If $\gamma+\frac{\alpha}{\lambda}$ is sufficiently small then any non-adaptive $(\Xi,\gamma)$-random privacy property tester with proximity parameter $\alpha$ has query complexity $\Omega(\max\{r, \frac{\lambda}{\alpha}\})$.
\end{thm}

We conjecture that the lower bound is actually $\Omega(r\frac{\lambda}{\alpha})$. If this is true then Theorem \ref{conversion} gives an almost optimal conversion from the two database setting to the random setting. 

\begin{proof} 
A random privacy property tester naturally induces a property tester in the two distribution setting $(P,Q)$ by setting $P_D=P$ for half the databases and $P_D=Q$ for the other half. Then $\{P_D\}$ is $(\Xi, \gamma)$-random private if $(P,Q)$ is $\Xi$ and $\alpha$-far if $(P,Q)$ is $\alpha$-far. Therefore, the random privacy tester must use at least as many queries as a privacy tester in the two database setting.

Suppose $(P,Q)$ is $\alpha$-far from $\Xi$-private and the data universe is uniformly distributed. If $\gamma+\frac{\alpha}{\lambda}$ is small enough then there exists a pair of nested subsets $S'\subset S\subset\mathbb{Z}^{\Omega}$ such that \[\mathbb{P}((S\times S^c)\cap \{(D,D')\;|\; D, D' \text{ neighbours }\})=\gamma+\frac{\alpha}{\lambda}\] and \[\mathbb{P}((S'\times S'^c)\cap \{(D,D')\;|\; D, D' \text{ neighbours }\})=\gamma.\] Define $P_D=P$ if $D\in S$, $P_D=Q$ if $D\in S^c$, $Q_D=P$ if $D\in S'$ and $Q_D=Q$ if $D\in S'$. Then $\{P_D\}$ is $(\Xi, \gamma)$-random private and $\{Q_D\}$ is $\alpha$-far from $(\Xi, \gamma)$-random private. 

Recall that a non-adaptive property testing algorithm can query by randomly sampling a pair of neighbours $D,D'$, and then sampling $P_D\times P_{D'}$. If $N$ is the normalisation factor, the distributions $\frac{1}{N}\sum_{(D,D') \text{ neighbours}} P_D\times P_{D'}$ and $\frac{1}{N}\sum_{(D,D') \text{ neighbours}} Q_D\times Q_{D'}$ have total variation distance $2\|P-Q\|_{TV}\mathbb{P}((S\backslash S')\times S^c)\cap \{(D,D')\;|\; D, D' \text{ neighbours }\})\le \frac{\alpha}{\lambda}$. Therefore, it takes at least $\frac{\lambda}{\alpha}$ queries to distinguish between $\{P_D\}$ and $\{Q_D\}$. 
\end{proof}

\section{No Information Setting}\label{noinfo}

We first show that no privacy property tester with finite query complexity exists for pDP. We then analyse a finite query complexity privacy property tester for aDP, as well query complexity lower bounds. For the remainder of this work we consider the two databases setting, where each algorithm $\mathcal{A}=(P_0,P_1)$ accepts only two databases, 0 and 1, as input and $\mathcal{A}(0)\sim P_0$ and $\mathcal{A}(1)=P_1$. The databases 0 and 1 are assumed to be neighbouring.

\subsection{Unverifiability}

The impossibility of testing pDP arises from the fact that very low probability events can cause the privacy parameters to increase arbitrarily. In each case we can design distributions $P$ and $Q$ that are close in TV-distance but for which the algorithm $(P,Q)$ has arbitrarily large privacy parameters. This intuition allows us to use a corollary of Le Cam's inequality (Corollary \ref{lecampropertytesting}) to prove our impossibility results.

\begin{lem}\label{lecampropertytesting}
For any privacy definition, let $\alpha$ be the proximity parameter and $\Xi$ be the privacy parameters. Suppose $\mathcal{A}=(P_0,P_1)$ and $\mathcal{B}=(Q_0, Q_1)$ are algorithms such that $\mathcal{A}$ is $\Xi$-DP and $\mathcal{B}$ is $\alpha$-far from being $\Xi$-DP. Then, any privacy property testing algorithm with QC $2r$ must satisfy \[\|(P_0^r\times P_1^r)-(Q_0^r\times Q_1^r)\|_{\text{TV}}\ge \frac{1}{3}\]
\end{lem}

%We are going to use Le Cam's inequality to prove Proposition \ref{approxminimax}. Let $\mathcal{P}$ be a class of distributions, $\theta:\mathcal{P}\to \Phi$ be the function we want to estimate, $\rho:\Phi\times\Phi\to\mathbb{R}_+$ a seminorm and $\varPhi:\mathbb{R}_+\to\mathbb{R}_+$ is a nondecreasing function s.t. $\varPhi(0)=0$. 

%For a distribution $P\in\mathcal{P}$, we assume we receive i.i.d. observations $X_i\sim P$ and \[\mathfrak{M}_m(\theta(\mathcal{P}, \varPhi\circ\rho) = \inf_{\hat{\theta}}\sup_{P\in\mathcal{P}}\mathbb{E}_P[\varPhi(\rho(\hat{\theta}(X_1, \cdots, X_r), \theta(P)))].\]

%\begin{lem}[Le Cam's]\label{lecams}
%If $\exists P_1, P_2\in\mathcal{P}$ s.t. $\rho(\theta(P_1), \theta(P_2))\ge 2\beta$ then 
%\[\mathfrak{M}_r(\theta(\mathcal{P}, \varPhi\circ\rho)\ge \varPhi(\beta)\left[\frac{1}{2}-\frac{1}{2}\|P_1^r-P_2^r\|_{TV}\right]\]
%\end{lem}

\iffalse
\begin{proof} Let $\alpha$ be the proximity parameter and $(\epsilon, \delta)$ be the privacy parameters. 
Suppose that $b(r)$ is the minimax rate for density estimation in $\mathcal{U}$. That is, $$\mathfrak{M}_r(\theta(\mathcal{P}, \varPhi\circ\rho)\ge b(r),$$ where $\varPhi$ is the identity, $\rho$ is $\|\cdot\|_{\text{TV}}$ and $\theta$ is density estimation. Thus, by Lemma~\ref{lecams}, there exists distributions $P,Q\in\mathcal{U}$ such that $\|P-Q\|_{\text{TV}}\ge \beta$ and $b(r)= \beta\left(\frac{1}{2}-\frac{1}{2}\|P^r-Q^r\|_{\text{TV}}\right)$. Therefore, $\|P^r-Q^r\|_{\text{TV}}= 1-\frac{4b(r)}{\beta}$. 
\end{proof}
\fi

\begin{thm}[pDP lower bound]\label{nogopure} 
Let $\alpha>0$ and $\epsilon>0$. No $\epsilon$-pDP property tester with proximity parameter $\alpha$ has finite query complexity.
\end{thm}

\begin{proof} Let $r$ be the query complexity of any pDP property tester.
Let $A>2\epsilon+\alpha$. Our goal is to prove that $r>\theta(e^A/A)$. If this is true for all $A$, the query complexity cannot be finite.

Consider algorithms, $\mathcal{A}=(P_0, P_1)$ and $\mathcal{B}=(Q_0, Q_1)$ where \[P_0=Q_0=e^{-A-\epsilon}\chi_{\psi}~+~(1~-~e^{-A-\epsilon})\chi_{\omega},\] \[P_1=e^{-A}\chi_{\psi}+(1-e^{-A})\chi_{\omega} \text{ and } Q_1=e^{-2A}\chi_{\psi}+(1-e^{-2A})\chi_{\omega}.\] Then $\mathcal{A}$ is $\epsilon$-pDP and $\mathcal{B}$ is $\alpha$-far from $\epsilon$-pDP. Now, by Pinsker's inequality,
\[\|(P_0^r, P_1^r), (Q_0^r, Q_1^r)\|_{\text{TV}}~\le~\sqrt{\frac{r}{2}}\sqrt{\text{D}_{\text{KL}}(P_0|Q_0)}.\] Therefore, by Lemma~\ref{lecampropertytesting}, 
\begin{align*}
r&\ge \frac{2}{9}\frac{1}{\text{D}_{\text{KL}}(P_0|Q_0)}\\
&=\frac{2}{9}\frac{1}{e^{-A}\log(e^{A})+(1-e^{-A})\log\left(\frac{1-e^{-A}}{1-e^{-2A}}\right)}=\theta\left(\frac{e^A}{A}\right).
\end{align*}
\end{proof}

We designed two distributions that are equal on a large probability set but for which the ratio $\frac{P_0(x)}{Q_0(x)}$ blows-up on a set with small probability. In Section \ref{fullinfo} we will see that testing pure DP becomes possible if we make assumptions on the algorithm $\mathcal{A}$. The assumption we need will ensure that $\frac{P_0(x)}{Q_0(x)}$ is upper bounded.

\subsection{Property Testing for aDP in the No Information Setting}

Fortunately, the situation is less dire for verifying aDP. Finite query complexity property testers do exist for aDP, although their query complexity can be very large. In the previous section, we relied on the fact that two distributions $P$ and $Q$ can be close in TV-distance while $(P,Q)$ has unbounded privacy parameters. In this section, we first show this is not true for aDP, which sets it apart from the other privacy notions. We then prove that the query complexity is $\Omega\left(\max\left\{n^{1-o(1)}, \frac{1}{\alpha^2}\right\}\right)$, and there exists an algorithm that uses  $O(\frac{4n(1+e^{2\epsilon})}{\alpha^2})$ queries per database. Define
\begin{equation}\label{boundary}
\delta^{\mathcal{A}}_{\epsilon} \ge \max_{D,D' \text{neighbours}}\max_{E} P_D(E)-e^{\epsilon}P_{D'}(E).
\end{equation}
An algorithm is $(\epsilon, \delta)$-aDP if and only if $\delta>\delta^{\mathcal{A}}_{\epsilon}$. The following lemma shows the relationship between the aDP parameters and TV-distance.

\begin{lem}\label{approxTV}
Let $\mathcal{A} = (P_0,P_1)$ and suppose $\mathcal{A}$ is $(\epsilon,\delta)$-aDP and $\alpha>0$. If $\mathcal{B}=(Q_0,Q_1)$ and
\begin{enumerate}
\item $\|P_0-Q_0\|_{\text{TV}}\le \alpha$ 
\item $\|P_1-Q_1\|_{\text{TV}}\le \alpha$,
\end{enumerate} 
then $\mathcal{B}$ is\ $(\epsilon, \delta+(1+e^{\epsilon})\alpha)-aDP.$ Furthermore, if $\alpha\le\frac{1-\delta}{1+e^{\epsilon}}$ then this bound is tight. That is, if $\delta^{\mathcal{A}}_{\epsilon}>0$, then there exists an algorithm $\mathcal{B}=(Q_0,Q_1)$ such that conditions (1) and (2) hold but $\mathcal{B}$ is $\alpha$-far from $(\epsilon, \delta^{\mathcal{A}}_{\epsilon}).$
%Conversely, assuming $\alpha\le \frac{e^{\epsilon}(1-\delta)-1}{2e^{2\epsilon}+e^{\epsilon}-1}$, there exists a pair of algorithms $\{P_0,P_1\}$ and $\{Q_0,Q_1\}$ such that conditions (1)-(3) hold and $\min_{\epsilon',\delta'} |\epsilon+\lambda\delta-\epsilon'-\lambda\delta'|\ge \lambda(\alpha+e^{\epsilon}\delta)$ where the minimum is over all $(\epsilon',\delta')$ such that $\{Q_0,Q_1\}$ is $(\epsilon',\delta')$-aDP.
\end{lem}

\begin{proof}
For any event $E$, \[Q_0(E)\le P_0(E)+\alpha\le e^{\epsilon}P_1(E)+\delta_{\mathcal{A}}+\alpha\le e^{\epsilon}Q_1(E)+e^{\epsilon}\alpha+\alpha+\delta_{\mathcal{A}}.\] Similarly, $Q_1(E)\le e^{\epsilon}Q_0(E)+e^{\epsilon}\alpha+\alpha+\delta_{\mathcal{A}}$. 

Conversely, let $\mathcal{A} = (P_0,P_1)$ and suppose $\delta^{\mathcal{A}}_{\epsilon}>0$. There must exist an event $E$ such that $P_0(E)=e^{\epsilon}P_1(E)+\delta^{\mathcal{A}}_{\epsilon}$. The condition on $\alpha$ can rewritten as $1-\alpha\ge e^{\epsilon}\alpha+\delta$ so we must have that either $P_0(E)\le 1-\alpha$ or $P_1(E)\ge\alpha$. 

First, suppose that $P_0(E)\le 1-\alpha$. Then there exists a distribution $Q_0$ such that \\$\|Q_0-P_0\|_{\text{TV}}=\alpha$ and $Q_0(E)=P_0(E)+\alpha$. If we let $Q_1=P_1$ then $Q_0(E)=e^{\epsilon}Q_1(E)+\alpha+\delta^{\mathcal{A}}_{\epsilon}$, which implies $\mathcal{B}=(Q_0,Q_1)$ is $\alpha$-far from $(\epsilon, \delta^{\mathcal{A}}_{\epsilon})$-aDP.

Finally, suppose $P_1(E)\ge\alpha$. Then there exists a distribution $Q_1$ such that \\$\|Q_1-P_1\|_{\text{TV}}=\alpha$ and $Q_1(E)=P_1(E)-\alpha$. Letting $Q_0=P_0$, again $\mathcal{B} = (Q_0,Q_1)$ is $\alpha$-far from $(\epsilon, \delta^{\mathcal{A}}_{\epsilon})$-aDP.
\end{proof}

\begin{thm}[Lower bound]\label{aDPLB}
Let $\alpha, \epsilon, \delta>0$ and suppose $e^{\epsilon}/2+\delta+\alpha<1$. Any $(\epsilon,\delta)$-aDP property tester with proximity parameter $\alpha$ has query complexity \[r~\ge~\max\left\{n^{1-o(1)}, \frac{1}{\alpha^2}\right\}.\] 
\end{thm}

The proof of the $1/\alpha^2$ component of the lower bound relies on Lemma \ref{lecampropertytesting} in a similar way to Theorem \ref{nogopure}. The proof of the $n^{1-o(1)}$ lower bound borrows a technique from \cite{Valiant:2011}. The lemma uses the fact that if two distributions only differ on elements of $[n]$ with low probability, then many samples are needed to distinguish between them. 

A property $\pi$ of a distribution is a function $\pi:\mathbb{Z}^{\Omega}\to\mathbb{R}$. It is called symmetric if for all permutations $\sigma$ and distributions $p$ we have $\pi(p)=\pi(p\circ\sigma)$. It is $(\alpha, \beta)$-weakly-continuous if for all distributions $p^+$ and $p^-$ satisfying $\|p^+-p^-\|_{\text{TV}}\le \beta$ we have $|\pi(p^+)-\pi(p^-)|\le \alpha$. The following lemma will be used in the proof of Theorem \ref{aDPLB}.

\begin{lem}\cite{Valiant:2011}\label{lowfreqbias}
Given a symmetric property $\pi$ on distributions on $[n]$ that is $(\alpha, \beta)$-weakly-continuous and two distributions, $p^+$, $p^-$ that are identical for any index occurring with probability at least $1/k$ in either distribution but where $\pi(p^+)>b$ and $\pi(p^-)<a$, then no tester can distinguish between $\pi>b-\epsilon$ and $\pi<q+\epsilon$ in $k\cdot\frac{\beta}{1000\cdot 2^{4\sqrt{\log n}}}$ samples.
\end{lem}

For aDP, our property is $\pi((P,Q))=\delta_{\epsilon}$, which is $((1+e^{\epsilon})\alpha, \alpha)$-weakly-continuous. We can now prove our lower bound. 

\begin{proof}[Proof of Theorem \ref{aDPLB}]
Let $P_0=Q_0$ be the uniform distribution on $\{\psi,\omega\}$. Let \[P_1 = \left(\frac{e^{\epsilon}}{2}+\delta\right)\chi_{\psi}+\left(1-\frac{e^{\epsilon}}{2}-\delta\right)\chi_{\omega}\] and \[Q_1 = \left(\frac{e^{\epsilon}}{2}+\delta+\alpha\right)\chi_{\psi}+\left(1-\frac{e^{\epsilon}}{2}-\delta-\alpha\right)\chi_{\omega}.\] Then, $(P_0,P_1)$ is $(\epsilon, \delta)$-aDP and $(Q_0,Q_1)$ is $\alpha$-far from $(\epsilon,\delta)$-aDP. Now, \[D_{\text{KL}}(P_1|Q_1) = \left(\frac{e^{\epsilon}}{2}+\delta+\alpha\right)\ln\left(1+\frac{\alpha}{\frac{e^{\epsilon}}{2}+\delta}\right)+\left(\frac{e^{\epsilon}}{2}-\delta-\alpha\right)\ln\left(1-\frac{\alpha}{\frac{e^{\epsilon}}{2}-\delta}\right)\lesssim \alpha^2.\] By the same argument as Theorem \ref{nogopure} we have $r = \Omega\left(\frac{1}{\alpha^2}\right)$.

Suppose $[n]$ is a disjoint union of the sets $R_1,R_2$ and $R_3$, all of which have cardinality $n/3$. Let $a=\frac{2\delta+\alpha}{3}, b=2a, \eta=\frac{\delta-\alpha}{3}$ so $a+\eta = \delta$ and $b-\eta=\delta+\alpha$. 
Let 
\begin{align*}
P_1=P_0=Q_0&=\frac{3a}{n}\chi_{R_1}&+\frac{3(1-a)}{n}\chi_{R_2}&&\\
Q_1&=&\frac{3(1-a)}{n}\chi_{R_2}&&+\frac{3a}{n}\chi_{R_3}.
\end{align*} 
Now, for $(P_0, P_1)$, $\delta_{\epsilon}\le a$ and for $(Q_0,Q_1)$, $\delta_{\epsilon}\ge 2a=b$. Since the distributions agree on any index with probability greater than $\frac{3a}{n}$, Lemma \ref{lowfreqbias} implies that no tester can distinguish between $\delta_{\epsilon}\ge b-\eta =\delta+\alpha$ and $\delta_{\epsilon}\le a+\eta=\delta$ with less than $\frac{n}{3a}(1+e^{\epsilon})\eta = \frac{3n(\delta-\alpha)}{2\delta+\alpha} =\Omega(n)$ samples.
\end{proof}

At first glance, Theorem \ref{aDPLB} doesn't look too bad. We should expect the sample complexity to scale like $1/\alpha^2$ since we need to have enough samples to detect the bad events. Our concern is the size of $\alpha$. If we would like $\alpha$ to be the same order as $\delta$, then our query complexity must scale as $\frac{1}{\delta^2}$. As we typically require $\delta$ to be extremely small (i.e. $\delta\approx 10^{-8}$), $\frac{1}{\delta}$ may be infeasibly large. If we are willing to accept somewhat larger $\delta$, then $\frac{1}{\delta^2}$ may be reasonable.

\begin{algorithm}[t]
   \caption{aDP Property Tester}
   \label{alg:adptest}
\begin{algorithmic}
   \STATE {\bfseries Input:} Universe size $n$, $\epsilon, \delta, \alpha>0$
   \STATE $\lambda=\max\{\frac{4n(1+e^{2\epsilon})}{\alpha^2}, \frac{12(1+e^{2\epsilon})}{\alpha^2}\}$
   \STATE Sample $r\sim $Poi$(\lambda)$
   \STATE Sample $D_0\sim P^r$, $D_1\sim Q^r$
   \FOR{$i\in[m]$}
   \STATE $x_i = $ number of $i$'s in $D_0$
   \STATE $y_i =$ number of $i$'s in $D_1$
   \STATE $z_i = \frac{1}{r}(x_i-e^{\epsilon}y_i)$
   \ENDFOR
   \STATE $z = \sum_{i=1}^r \max\{0,z_i\}$
   \IF{$z<\delta+\alpha$}
   \STATE {\bfseries Output:} ACCEPT
   \ELSE
   \STATE {\bfseries Output:} REJECT
   \ENDIF
\end{algorithmic}
\end{algorithm}

We now turn our attention to Algorithm \ref{alg:adptest}, a simple algorithm for testing aDP with query complexity $O(\frac{4n(1+e^{2\epsilon})}{\alpha^2})$. Its sample complexity matches the lower bound in Theorem \ref{aDPLB} in $\alpha$ when $n$ is held constant and in $n$ when $\alpha$ is held constant. We are going to use a trick called \emph{Poissonisation} to simplify the proof of soundness and completeness, as in \cite{Batu:2013}. Suppose that, rather than taking $r$ samples from $P$, the algorithm first samples $r_1$ from a Poisson distribution with parameter $\lambda=r$ and then takes $r_1$ samples from $P$. Let $X_i$ be the random variable corresponding to the number of times the element $i\in[n]$ appears in the sample from $P$. Then $X_i$ is distributed identically to the Poisson distribution with parameter $\lambda=p_ir$ and all the $X_i$'s are mutually independent. Similarly, we sample $r_2$ from a Poisson distribution with parameter $r$ and then take $r_2$ samples from $Q$. Let $Y_i$ be the the number of times $i$ appears in the sample from $Q$, so $Y_i$ is Poisson with $\lambda=q_ir$ and the $Y_i$ are independent.

\begin{thm}[Upper bound]\label{aDPUB}
Let $\epsilon,\delta,\alpha>0$. Algorithm \ref{alg:adptest} is a $(\epsilon, \delta)$-aDP property tester with proximity parameter $2\alpha$ and sample complexity $O(\frac{4n(1+e^{2\epsilon})}{\alpha^2})$. 
\end{thm}

\begin{proof} Let $p_i=P(i)$ and $q_i=Q(i)$.
Let $Z_i = \frac{1}{r}\left(X_i-e^{\epsilon}Y_i\right)$ so $\mathbb{E}[Z_i]=p_i-e^{\epsilon}q_i$ and Var$(Z_i)\le \frac{p_i+e^{2\epsilon}q_i}{r}$. Note also that $(P,Q)$ is $(\epsilon, \delta)$-DP if $\Delta := \sum_{i=1}^n\max\{0,\mathbb{E}[Z_i]\}\le \delta$. First note that $\mathbb{E}[Z]\ge\Delta$.
If $\mathbb{E}[Z_i]\le0$ then \begin{align*}
\mathbb{E}[\max\{0,Z_i\}] &= \int_0^{\infty}\mathbb{P}(\max\{0,Z_i\}> x) dx\\
&= \int_0^{\infty}\mathbb{P}(Z_i> x) dx\\
&\le \int_0^{\infty}\mathbb{P}(Z_i-\mathbb{E}[Z_i]> x-\mathbb{E}[Z_i]) dx\\
&\le \int_0^{\infty}\min\{1, \frac{VarZ_i}{(x-\mathbb{E}[Z_i])^2}\}dx\\
&= \int_0^{\sqrt{VarZ_i}+\mathbb{E}[Z_i]}1dx+\int_{\sqrt{VarZ_i}+\mathbb{E}[Z_i]}^{\infty}\frac{VarZ_i}{(x-\mathbb{E}[Z_i])^2}dx\\
&= \sqrt{VarZ_i}+\mathbb{E}[Z_i]+\sqrt{VarZ_i}\\
&\le \sqrt{\frac{p_i+e^{2\epsilon}q_i}{r}}
\end{align*}
If $\mathbb{E}[Z_i]>0$ then $\mathbb{E}[\max\{0,X_i\}]\le \mathbb{E}[Z_i]+\sqrt{\frac{p_i+e^{2\epsilon}q_i}{r}} = p_i-e^{\epsilon}q_i+\sqrt{\frac{p_i+e^{2\epsilon}q_i}{r}}$. Therefore, \[\mathbb{E}[Z] \le \Delta + \sum_{i=1}^n \sqrt{\frac{p_i+e^{2\epsilon}q_i}{r}}\le \Delta+\sqrt{\frac{n}{r}}(1+e^{2\epsilon})\]
Now, let $Z_i'$ be an independent copy of $Z_i$ then 
\begin{align*}
Var[\max\{0,Z_i\}] &= \frac{1}{2}\mathbb{E}[(\max\{0,Z_i\}-\max\{0,Z_i'\})^2]\\
&= \int_0^{\infty} \mathbb{P}((\max\{0,Z_i\}-\max\{0,Z_i'\})^2\ge x)dx\\
&\le \int_0^{\infty} \mathbb{P}((Z_i-Z_i')^2\ge x)dx\\
&= VarZ_i = \frac{p_i+e^{2\epsilon}q_i}{r}.
\end{align*}
So $Var Z\le \frac{1+e^{2\epsilon}}{r}$. Therefore, 
\begin{align*}
\mathbb{P}(|Z-\Delta|\ge \alpha) &\le \mathbb{P}(|Z-\mathbb{E}[Z]|+|\mathbb{E}[Z]-\Delta|\ge \alpha)\\
&\le \mathbb{P}(|Z-\mathbb{E}[Z]|\ge \alpha-\sqrt{\frac{n}{r}}(1+e^{2\epsilon}))\\
&\le \frac{VarZ}{(\alpha-\sqrt{\frac{n}{r}}(1+e^{2\epsilon}))^2}\\
&\le \frac{1+e^{2\epsilon}}{r(\alpha-\sqrt{\frac{n}{r}}(1+e^{2\epsilon}))^2},
\end{align*}
which is less than $1/3$ if $r\ge\max\{\frac{4n(1+e^{2\epsilon})^2}{\alpha^2}, \frac{12(1+e^{2\epsilon})}{\alpha^2}\}$.
\end{proof}

\section{Full Information (FI) Setting}\label{fullinfo}

The situation is substantially rosier if we have side-information. Although there are some realistic scenarios where one may have \emph{trusted} side-information, we will focus on \emph{untrusted} side-information. In particular, we allow our property tester to REJECT simply because the provided side-information is inaccurate. We will see that the untrusted side-information can still be useful since verifying information is often easier than estimating it. 

The usefulness of side-information in property testing is informally lower bounded by how easy it is to generate the same information, and how much the information tells us about the property. Proposition \ref{meannogood} below demonstrates this idea when the side information is the means of the distributions. The means of the distributions $(P_0,P_1)$ are efficient to estimate, but do not tell us very much about whether or not the privacy guarantee is satisfied. If $\mathcal{A}$ is an unbiased estimate for a function $f$, then the following proposition states that \emph{knowing the function the black-box is computing does not help in verifying pDP}. 

\begin{prop}\label{meannogood}
Let $\alpha, \epsilon>0$ and suppose the side information is $(a,b)$, which are purported to be the means of $P_0$ and $P_1$. If there exists a $\Xi$-private algorithm $(P_0,P_1)$ supported on $[n-1]$ with $\mathbb{E}(P_0)=a$ and $\mathbb{E}(P_1)=b$, then no $\epsilon$-pDP property tester with side information $(a, b)$ and proximity parameter $\alpha$ has finite query complexity.
\end{prop}

The requirement that a $\Xi$-private algorithm with the right side information exists is necessary. If no such algorithm exists, then the tester should always REJECT and requires no queries. Under the assumption that such an algorithm exists, it is reasonable to assume that there exists an algorithm with slightly smaller support. 

\begin{proof} 
Let $A>0$ and $\mathcal{A}=(P_0,P_1)$ be the algorithm promised. That is, $(P_0,P_1)$ is $\Xi$-private and supported on $[n-1]$. Let $Q_0=(1-e^{-A})P_0+e^{-A}\chi_n$ and $Q_1=P_1$ so $(Q_0,Q_1)$ is $\alpha$-far from pDP. Now, $\|P_0\times P_1-Q_0\times Q_1\|_{\text{TV}}\le e^{-A}$ so it requires $e^{A}\to\infty$ samples to distinguish between the two distributions.
\end{proof}

For the remainder of the paper we focus on what we call the \emph{full information} setting: we are given sample access to $\mathcal{A}$ and a distribution $Q_D$ for each database $D$. This case may seem optimistic, however we will find that the lower bounds obtained are still very large. The lower bounds in this optimistic setting also hold for more realistic setting when the verifier has less information. In contrast to the mean, this side information is very informative about the privacy of the algorithm. It is also difficult to generate based on samples. We can \emph{estimate} it using $\Theta(\frac{n}{\alpha^2})$ \citep{Chan:2014} queries of each database, where $\alpha$ is the accuracy in TV-distance. However, we already know that the only privacy notion for which an estimate is sufficient is aDP. 

In Algorithm \ref{alg:FIaDP} we use an identity tester rather than density estimation to obtain a lower sample complexity. An identity tester is a property tester $T$ that takes as input a description of the discrete distribution $P$ and $m$ samples from a distribution $Q$. If $P=Q$ then the tester ACCEPTS with probability at least 2/3 and if $\|P-Q\|_{\text{TV}}\ge\alpha$ then the tester REJECTS with probability at least 2/3. It is also known that testing identity to a known distribution requires asymptotically less samples than estimating an unknown distribution. 

\begin{algorithm}[t]
   \caption{aDP FI Property Tester}
   \label{alg:FIaDP}
\begin{algorithmic}
   \STATE {\bfseries Input:} Universe size $n$, $\epsilon, \delta, \alpha>0$, $(Q_0,Q_1)$ and identify tester $T$ with sample complexity $r$
   \IF{$(Q_0,Q_1)$ is not $(\epsilon, \delta)$-aDP}
   \STATE {\bfseries Output:} REJECT
   \ELSE
   \IF{$T(Q_0,x\sim P_0^r)=$ REJECT or $T(Q_1,x\sim P_1^r)=$ REJECT}
   \STATE {\bfseries Output:} REJECT
   \ELSE
   \STATE {\bfseries Output:} ACCEPT
   \ENDIF
   \ENDIF	
   \end{algorithmic}
\end{algorithm}

\begin{prop}\label{FIaDP}
There exists a identity tester $T$ such that Algorithm \ref{alg:FIaDP} is a $(\epsilon, \delta)$-aDP FI property tester with query complexity $O\left(\frac{\sqrt{n}}{\alpha^2}\right)$ and proximity parameter $\alpha$.
\end{prop}

This is our first, and only, sublinear query complexity property tester for privacy. Since closeness in TV-distance implies closeness in aDP, we only need to check that the true distributions are close to $(Q_0,Q_1)$ and that $(Q_0,Q_1)$ is $(\epsilon, \delta)$-aDP. The difficult part is testing closeness of the distributions, for which we borrow from \cite{On:2014}. 

\begin{proof} 
\cite{On:2014} proved that there exists a property tester $T$ that takes as input a description of the discrete distribution $P$ and $O(\frac{\sqrt{n}}{\alpha^2})$ samples from an distribution $Q$. If $P=Q$ then the tester ACCEPTS with probability at least 2/3 and if $\|P-Q\|_{\text{TV}}\ge\alpha$ then the tester REJECTS with probability at least 2/3. If we increase the sample complexity of the tester $T$ by a constant factor then we can replace $2/3$ with $\sqrt{2/3}$.

To prove completeness, suppose $(P_0,P_1)=(Q_0,Q_1)$ and $(Q_0,Q_1)$ is $(\epsilon, \delta)$-aDP. Since $T$ ACCEPTS on both pairs $(Q_0,P_0)$ and $(Q_1,P_1)$ with probability at least $\sqrt{2/3}$, it ACCEPTS both with probability at least $2/3$. To prove soundness, suppose $(P_0,P_1)$ is\\ $(1~+~e^{\epsilon})\alpha$-far from $(\epsilon, \delta)$-aDP. Assume $(Q_0,Q_1)$ is $(\epsilon, \delta)$-aDP because otherwise the tester REJECTS. It must be the case that either $\|Q_0-P_0\|_{\text{TV}}\ge\alpha$ or $\|Q_1-P_1\|_{\text{TV}}\ge\alpha$ because otherwise $(P_0,P_1)$ would be $(\epsilon, (1+e^{\epsilon})\alpha)$-aDP by Lemma \ref{approxTV}. Thus, with probability at least $2/3$ either $T(Q_0,x\sim P_0)=$ REJECT or $T(Q_1,x\sim P_1)=$ REJECT.
\end{proof}

Next, we show that for pDP, the side information allows us to obtain a finite query complexity property tester. The side-information gives us an easy way to switch from a worst-case analysis to input specific upper bounds. We argue that \[\beta = \min_E\min_D P_D(E),\] where the first $\min$ is over events $E$ and the second is over databases $D$, is the crucial quantity in understanding verifiability in the full information setting. 
 
The lower bound proofs in the previous section all proceeded by finding two algorithms $\mathcal{A}$ and $\mathcal{B}$ that were close in TV-distance but had very different privacy parameters. The algorithms we chose all had one feature in common: the distributions $P_D$ contained very low probability events. This property allowed us to drive the denominator of $\frac{Q_D(E)}{P_D(E)}$ to 0, and hence the privacy loss to $\infty$, while remaining close in TV-distance. This method works equally well in the full-information setting \emph{if} low probability events exist in the distributions $Q_D$. 

If the distribution $Q_D$ does not have low probability events, then any distributions close to $P_D$ must have bounded privacy parameters. To see this, suppose $(Q_0, Q_1) = (U,U)$ where $U$ is the uniform distribution $U$ on $\{\psi,\omega\}$. We can establish in approximately $\frac{1}{\alpha^2}$ samples whether or not $P_0$ and $P_1$ are both within TV-distance $\alpha$ of uniform. If not, then we REJECT. If so, then the worst case for privacy is $P_0=(1/2-\alpha)\chi_{\psi}+(1/2+\alpha)\chi_{\omega}$ and $P_1=(1/2+\alpha)\chi_{\psi}+(1/2-\alpha)\chi_{\omega}$. However, the increase in the pDP parameter from $(Q_0, Q_1)$ to $(P_0,P_1)$ is bounded by $\ln\frac{1/2+\alpha}{1/2-\alpha}\approx \alpha$.

Algorithm~\ref{alg:pDPFI} is a full information property tester for pDP. Note that this algorithm is not sublinear in $n$ since $\beta<\frac{1}{n}$.

\begin{algorithm}[t]
   \caption{pDP FI Property Tester}
   \label{alg:pDPFI}
\begin{algorithmic}
   \STATE {\bfseries Input:} Universe size $n$, $\epsilon, \alpha>0$, $(Q_0, Q_1)$
   \STATE $\lambda=\frac{\ln n}{\alpha^2\beta^2}$
   \STATE Sample $r\sim $Poi$(\lambda)$
   \STATE Sample $D_0\sim P_0^r$, $D_1\sim P_1^r$
   \FOR{$i\in[m]$}
   \STATE $x_i = $ number of $i$'s in $D_0$
   \STATE $y_i =$ number of $i$'s in $D_1$
   \ENDFOR
   \STATE $\hat{\epsilon} = \sup_i \max\{\ln\frac{x_i}{y_i}, \ln\frac{y_i}{x_i}\}$
   \IF{$\hat{\epsilon}>\epsilon+2\alpha$}
   \STATE {\bfseries Output:} REJECT
   \ELSE
   \IF{$\forall i \;\; e^{-\alpha}\le \frac{x_i}{(Q_0)_i}\le e^{\alpha}$ and $e^{-\alpha}~\le~\frac{y_i}{(Q_1)_i}~\le~e^{\alpha}$}
   \STATE {\bfseries Output:} ACCEPT
   \ELSE 
   \STATE {\bfseries Output:} REJECT
   \ENDIF
   \ENDIF
\end{algorithmic}
\end{algorithm}

\begin{thm}[pDP upper bound]\label{FIpureUB}
Let $\epsilon>0$ and $\alpha>0$. Algorithm \ref{alg:pDPFI} is an $\epsilon$-aDP FI property tester with proximity parameter $10\alpha$ and query complexity $O\left(\frac{\ln n}{\alpha^2\beta^2}\right).$ 
\end{thm}

\begin{proof}
We first show completeness. Suppose $(P_0,P_1)=(Q_0, Q_1)$ and $\mathcal{A}$ is $\epsilon$-pDP. By the multiplicative Hoeffding's inequality, \[\mathbb{P}\left(\frac{x_i}{(P_0)_i} \ge e^{\alpha} \text{ or } \frac{(P_0)_i}{x_i} \ge e^{\alpha}\right)\le e^{-2r(e^{\alpha}-1)^2\beta^2}+e^{-2r(1-e^{-\alpha})^2\beta^2}.\] Therefore, \[\mathbb{P}\left(\exists i \text{ s.t. }\;\;\frac{x_i}{(P_0)_i} \ge e^{\alpha} \text{ or } \frac{(P_0)_i}{x_i} \ge e^{\alpha}\right)\le n(e^{-2r(e^{\alpha}-1)^2\beta^2}+e^{-2r(1-e^{-\alpha})^2\beta^2})\le \frac{1}{6}.\]
Thus with probability $2/3$ we have for all $i$, $e^{-\alpha}\le \frac{x_i}{(P_0)_i}\le e^{\alpha}$ and $e^{-\alpha}\le \frac{y_i}{(P_1)_i}\le e^{\alpha}$ and so \[\frac{x_i}{y_i} = \frac{x_i}{(P_0)_i}\frac{(P_0)_i}{(P_1)_i}\frac{(P_1)_i}{y_i}\le e^{\alpha}e^{\epsilon}e^{\alpha}.\] This implies $\hat{\epsilon}\le\epsilon+2\alpha$ so we ACCEPT.

For soundness, we show that the ACCEPT conditions imply that $(P_0,P_1)$ must be at least $(\epsilon~+~10\alpha)$-pDP. The condition $e^{-\alpha}\le \frac{x_i}{(Q_0)_i}\le e^{\alpha}$ implies $|x_i-(Q_0)_i|\le \max\{(e^{\alpha}-1)(Q_0)_i, (1-e^{-\alpha})(Q_0)_i\}$. Also, by the additive Hoeffding's inequality 
\begin{align*}
\mathbb{P}\Big(\exists i \text{ s.t. } \; |x_i-(P_0)_i|\ge (Q_0)_i\max&\{(e^{\alpha}-1), (1-e^{-\alpha})\}\Big)\\
&\le n\min\{e^{-r(e^{\alpha}-1)^2\beta^2}, e^{-r(1-e^{-\alpha})^2\beta^2}\}\\
&\le \frac{1}{6}.
\end{align*}
Therefore, with probability $2/3$, \[|(P_0)_i-(Q_0)_i|\le (Q_0)_i\max\{2(e^{\alpha}-1), 2(1-e^{-\alpha})\}\le 2\alpha(Q_0)_i \] for sufficiently small $\alpha$. This implies $\max\{\frac{(P_0)_i}{(Q_0)_i}, \frac{(Q_0)_i}{(P_0)_i}\}\le e^{4\alpha}$. Similarly, \[\max\{\frac{(P_1)_i}{(Q_1)_i}, \frac{(Q_1)_i}{(P_1)_i}\}~\le~e^{4\alpha}.\] Since $\hat{\epsilon}\le \epsilon+2\alpha$ we have \[\frac{(P_0)_i}{(P_1)_i} = \frac{(P_0)_i}{(Q_0)_i}\frac{(Q_0)_i}{x_i}\frac{x_i}{y_i}\frac{y_i}{Q_1)_i}\frac{(Q_1)_i}{(P_1)_i}\le e^{4\alpha+\alpha+\epsilon+2\alpha+\alpha+4\alpha}.\] 
\end{proof}

We now turn to lower bounding the query complexity of aDP testing in the FI setting. The sample complexity is tight in $\alpha$ but deviates by a factor of $\beta$.

\begin{thm}[pDP lower bound] \label{FILB}
Let $\alpha>0$ and $\ln2>\epsilon>0$. Given side information $(Q_0,Q_1)$, any $\epsilon$-pDP property tester with proximity parameter $\alpha$ has query complexity $\Omega\left(\frac{1}{\beta\alpha^2}\right).$ 
\end{thm}

\begin{proof} Let $\psi, \omega, \phi\in[n]$ and notice that $\beta<1/2$ provided $n>2$.
To prove the lower bound let \[Q_0= \beta\chi_{\psi}+\beta\chi_{\omega}+(1-2\beta)\chi_{\phi} \text{   and   } Q_1=e^{\epsilon}\beta\chi_{\psi}+(2-e^{\epsilon})\beta\chi_{\omega}+(1-2\beta)\chi_{\phi}\] be the side-information. Then $(Q_0,Q_1)$ is $\epsilon$-pDP. Let \[P_0=e^{-\alpha}\beta\chi_{\psi}+(2-e^{-\alpha})\beta\chi_{\omega}+(1-2\beta)\chi_{\phi} \text{   and   } P_1=Q_1\] so $(P_0,P_1)$ is $\alpha$-far from $\epsilon$-pDP. Now, 
\begin{align*}
D_{\text{KL}}(P_0|Q_0) &= \beta\ln\frac{\beta}{e^{-\alpha}\beta}+\beta\ln\frac{\beta}{(2-e^{-\alpha})\beta} = \beta\alpha+\beta\ln\left(1-\frac{1-e^{-\alpha}}{2-e^{-\alpha}}\right)\\
&\le \beta\alpha-\beta\alpha+\beta\alpha^2=\beta\alpha^2.
\end{align*}
As in Theorem \ref{nogopure}, we must have 
\begin{align*}
r &\ge \frac{2}{9}\frac{1}{D_{\text{KL}}(P_0|Q_0)} = \Omega\left(\frac{1}{\beta\alpha^2}\right).
\end{align*}
\end{proof}

% Acknowledgments---Will not appear in anonymized version
%\acks{We thank a bunch of people.}

\bibliographystyle{plainnat}
\bibliography{bibliography}

\end{document}